\documentclass[tpa,11pt,preprint,letterpaper]{myimsart} \pubyear{January 15, 2014}

\RequirePackage[OT1]{fontenc}
\RequirePackage{amsthm,amsmath}
\RequirePackage[numbers]{natbib}
\RequirePackage[colorlinks,citecolor=blue,urlcolor=blue]{hyperref}
\usepackage{graphicx,float,latexsym,times}
\usepackage{amsfonts,amstext,amssymb}
\usepackage{mathrsfs}
\usepackage{bbm}

\usepackage{epstopdf}

\volume{}
\issue{}
\firstpage{}
\lastpage{}
\arxiv{}

\startlocaldefs
\numberwithin{equation}{section}
\theoremstyle{plain}
\newtheorem{theorem}{Theorem}[section]

\newcommand{\ind}[1]{\mathbbm{1}_{\{#1\}}}   

\newcommand{\ccF}{\mathscr{F}}
\newcommand{\ccJ}{\mathscr{J}}

\newcommand{\ccG}{\mathscr{G}}
\newcommand{\ccA}{\mathscr{A}}

\newcommand{\cA}{\mathcal{A}}
\newcommand{\cJ}{\mathcal{J}}

\newcommand{\Real}{\mathbb{R}}

\newcommand{\ignore}[1]{{}}

\newcommand{\esinf}{\text{ess}\,\!\inf}
\newcommand{\esup}{\text{ess}\,\!\sup}
\newcommand{\Exp}{{\sf E}}
\newcommand{\Pro}{{\sf P}}

\newcommand{\TS}{\mathcal{S}}
\newcommand{\TC}{T_{\text{C}}}

\endlocaldefs

\begin{document}

\begin{frontmatter}
\title{Multiple optimality properties of the\\ Shewhart test}
\runtitle{Multiple optimality properties of the Shewhart test}

\begin{aug}
\author{\fnms{George V.} \snm{Moustakides}\ead[label=e1]{moustaki@upatras.gr}}
\affiliation{University of Patras}
\address{Department of Electrical and\\ Computer
Engineering\\
University of Patras\\ 26500 Rio\\ Greece\\ \printead{e1}}
\runauthor{GEORGE V. MOUSTAKIDES}
\end{aug}

\begin{abstract}
For the problem of sequential detection of changes, we adopt the probability maximizing approach in place of the classical minimization of the average detection delay, and propose modified versions of the Shiryaev, Lorden and Pollak performance measures. For these alternative formulations, we demonstrate that the optimum sequential  detection scheme is the simple Shewhart rule. Interestingly, we can also solve problems which under the classical setup have been open for many years, as optimum change detection with time varying observations or with multiple post-change probability measures. For the last case, we also offer the exact solution for Lorden's original setup when the average false alarm period is within certain limits.
\end{abstract}



\end{frontmatter}

\section{Introduction}
Suppose $\{\xi_t\}_{t>0}$ is a discrete-time process which becomes available sequentially and define $\{\ccF_t\}_{t\geq0}$ to be the associated filtration with $\ccF_t=\sigma\{\xi_1,\ldots,\xi_t\}$ the $\sigma$-algebra generated by the observations up to time $t$. Let $\tau\in\{\ldots,-1,0,1,\ldots\}$ denote a \textit{changetime} and assume that the observations follow the probability measure $\Pro_\infty$ up to and including $\tau$, while after $\tau$ the probability measure switches to $\Pro_0$. If the change in statistics takes place at $\tau=t$ then this induces a probability measure which we denote with $\Pro_t$ while $\Exp_t[\cdot]$ is reserved for the corresponding expectation. We would like to stress that, here, $\tau$ denotes the \textit{last} time instant under the nominal regime and \textit{not} the first under the alternative which is the usual practice. This slight difference allows to view $\tau$ as a \textit{stopping time} (the time the observations \textit{stop} following the nominal statistics), property which can be analytically very convenient (see Moustakides \cite{Moustakides2}).

We are interested in detecting the occurrence of the changetime $\tau$ with the help of a stopping time $T$ adapted to the filtration $\{\ccF_t\}$ that will signal the change as soon as possible avoiding, at the same time, making frequent false alarms. The effectiveness of a detection scheme is commonly quantified through the average detection delay. There are, of course, various possibilities depending on the prior knowledge we have and the model we adopt for the changetime. In particular, assuming $\tau$ to be random, independent from the observations, with a known prior, Shiryaev \cite{Shiryaev} proposed the following measure
\begin{equation}
\cJ_{\rm S}(T)=\Exp[T-\tau|T>\tau].
\label{eq:shiryaev}
\end{equation}
If we consider $\tau=t$ to be deterministic and unknown we can then follow a worst-case analysis and consider the performance measure proposed by Lorden \cite{Lorden}
\begin{equation}
\cJ_{\rm L}(T)=\sup_{t\ge0}\esup\Exp_t[T-t|\ccF_t,T>t].
\label{eq:lorden}
\end{equation}
Finally, assuming again that $\tau=t$ is deterministic and unknown we can alternatively define 
\begin{equation}
\cJ_{\rm P}(T)=\sup_{t\geq0}\Exp_t[T-t|T>t],
\label{eq:pollak}
\end{equation}
which is the criterion introduced by Pollak \cite{Pollak}. The three measures depicted in \eqref{eq:shiryaev},\eqref{eq:lorden},\eqref{eq:pollak} are the most common criteria encountered in the literature and, as noted in \cite{Moustakides2}, they can be recovered from a general definition that treats $\tau$ as a stopping time. An optimum stopping rule $T$ is then specified by minimizing these performance measures subject to suitable false alarm constraints.

\subsection{Criteria based on detection probability} \label{ssec:1.0}
We observe from \eqref{eq:shiryaev},\eqref{eq:lorden}, \eqref{eq:pollak} that no hard limit is imposed on the detection delay. Consequently, this quantity can become arbitrarily large. As reported in Gu\'epi\'e et al. \cite{Nikiforov} and in references therein there are several applications in practice where unbounded delays can be undesirable and one would rather detect the change within a pre-specified time window, \textit{after} the change has occurred\footnote{According to our definition, stopping before and at $\tau$ corresponds to false alarm.}. In other words we like to have $\tau<T\leq\tau+m$, for given $m\geq1$. Stopping within the prescribed interval constitutes a desirable event while if $T>\tau+m$ this is \textit{not} considered as successful detection.

Similarly to \eqref{eq:shiryaev},\eqref{eq:lorden},\eqref{eq:pollak}, we can now propose the following alternatives of the three classical performance measures
\begin{align}
\ccJ_{\rm S}(T)&=\Pro(\tau<T\leq\tau+m|T>\tau) \label{eq:shiryae}\\
\ccJ_{\rm L}(T)&=\inf_{t\ge0}\esinf\Pro_t(t<T\leq t+m|\ccF_t,T>t)\label{eq:lorde}\\
\ccJ_{\rm P}(T)&=\inf_{t\ge0}\Pro_t(t<T\leq t+m|T>t)\label{eq:polla}.
\end{align}
As we can see, instead of focusing on the average detection delay, we now pay attention to the detection probability. Consequently, here, we need to replace the minimization of the worst-case average detection delay of the classical approach with the \textit{maximization of the worst-case detection probability}. 

Bojdecki \cite{Bojdecki} was the first to adopt this probability maximizing idea by considering the maximization of the probability $\Pro(|\tau+1-T|\leq M)$. The complete solution to this problem was offered for the case $M=0$ and for the Bayesian formulation with the changetime $\tau$ following a geometric prior. The optimum stopping time turned out to be the simple test introduced by Shewhart in \cite{Shewhart} and which will also become our main focus in the analysis that follows.
We should mention that $M=0$ corresponds to the maximization of the probability of the event $\{T=\tau+1\}$, namely that detection is achieved by using just \textit{the first observation under the alternative regime}. A point we need to make is that Bojdecki, in his approach, did not attempt to control false alarms in any sense. Following similar ideas, Sarnowski and Szajowski \cite{Szajowski} extended this result to the dependent observations case; while very recently Pollak and Krieker \cite{Pollak2} considered the i.i.d.~case but with the data after the change distributed according to a parametric family of pdfs and the parameters following a known prior. Pollak and Krieker \cite{Pollak2} also adopted a semi-Bayesian approach where the changetime $\tau$ is deterministic and unknown while the post-change density, as before, is a parametric family with the parameters distributed according to a known prior. In the current work, unlike \cite{Bojdecki} and \cite{Szajowski}, we follow the common practice of the classical formulation and, as in the semi-Bayesian approach of \cite{Pollak2}, we impose suitable constraints for false alarm control. 

Before continuing with the detailed presentation of the various formulations, we first recall the form of the Shewhart test \cite{Shewhart} that we are going to adopt for our analysis. Consider observations $\{\xi_t\}$ that are independent but not necessarily identically distributed before and after the change and denote with $\{\ell_t\}$ the corresponding sequence of likelihood ratios. We are then interested in the following form of the Shewhart test\footnote{\label{foot:3} To avoid unnecessary technical complications, throughout our work, we are going to assume that the cdfs of all likelihood ratios $\ell_t$, under both probability measures, are continuous and strictly increasing functions.}
\begin{equation}
\TS=\inf\{t>0:~\ell_t\ge\nu_t\}.
\label{eq:Shewhart}
\end{equation}
The threshold sequence $\{\nu_t\}$ is deterministic and its exact form depends on the criterion we adopt and the statistics of the observations.

Having defined the Shewhart stopping time of interest, we briefly recall an optimality result for this test which has already been established in Moustakides \cite{Moustakides1}. In particular, in the next subsection we discuss the fact that the Shewhart test \textit{matches} CUSUM as long as the average false alarm period does not exceed a specific value.

\subsection{Optimality with respect to Lorden's classical criterion}\label{ssec:1.1}
In the case of i.i.d.~observations before and after the change with corresponding pdfs $f_\infty(\xi)$ and $f_0(\xi)$, in \cite{Moustakides1} it was proved that CUSUM solves the following constrained optimization problem proposed by Lorden \cite{Lorden}
\begin{multline}
\inf_T\sup_{t\ge0}\,\esup\,\Exp_t[T-t|\ccF_t,T>t];~~\text{over all}~T:~\Exp_\infty[T]\geq\gamma\geq1.
\label{eq:Lorden_opt}
\end{multline}
The CUSUM stopping time $\TC$ is defined in \cite{Moustakides1} as follows: For $t>0$ let
\begin{equation}
Y_t=\max\{Y_{t-1},1\}{\mit\ell_t},~Y_0=0;~~\TC=\inf\{t>0:~Y_t\ge\nu\},
\label{eq:CUSUM}
\end{equation}
where $\{Y_t\}$ is the CUSUM statistic, while the constant threshold $\nu\ge0$ is selected so that the false alarm constraint is satisfied with equality.

It is interesting to note that, customary, the CUSUM statistic $Y_t$ is specified in the literature slightly differently, namely, $Y_t=\max\{Y_{t-1}\ell_t,1\}$. When $\nu>1$, the two statistics give rise to exactly the same stopping time $\TC$. However, when $0\leq\nu\leq1$, by adopting the classical definition we are forced to stop at $T=1$, while \eqref{eq:CUSUM} results in a nontrivial stopping time. It is in fact for these values of the threshold, that is, $\nu\in[0,1]$ that CUSUM is reduced to the Shewhart test. Indeed note from \eqref{eq:CUSUM} that, as long as we do not stop at $t-1$, we have $Y_{t-1}<\nu$. Consequently when $\nu\leq1$ this immediately implies that $Y_t=\max\{Y_{t-1},1\}\ell_t=\ell_t$ suggesting that CUSUM is reduced to the Shewhart rule \eqref{eq:Shewhart} with constant threshold $\nu$.

Let us identify the range of false alarm rates for which CUSUM is equivalent to the Shewhart test.
Since under each probability measure the sequence $\{\ell_t\}$ is i.i.d.~we can conclude
\begin{multline}
\Exp_i[\TS]=\sum_{t=0}^\infty\Pro_i(\TS>t)=\sum_{t=0}^\infty\Pro_i(\ell_{1}<\nu;\ell_{2}<\nu;\cdots;\ell_{t}<\nu)\\
=\sum_{t=0}^\infty\big(\Pro_i(\ell_1<\nu)\big)^t=\frac{1}{\Pro_i(\ell_1\ge\nu)},~~i=0,\infty.
\label{eq:Average}
\end{multline}
From the previous equality we deduce that the largest value of the false alarm rate $\gamma$, for which $\TC=\TS$, is achieved when $\nu=1$. This implies that for $\gamma\in[1,\Pro_\infty^{-1}(\ell_1\ge1)]$ we can find a threshold $\nu\in[0,1]$ such that CUSUM is reduced to the Shewhart rule. It is also clear that the classical definition of CUSUM cannot accommodate any false alarm rate within the same interval. 

The previous range of false alarm rates can become more pronounced if we consider the exponential penalty criterion proposed by Poor \cite{Poor}, that is,
$$
\hat{\cJ}_{\rm L}(T)=\sup_{t\ge0}\,\esup\,\Exp_t\left[\frac{1-c^{T-t}}{1-c}\,\Big|\,T>t,\ccF_t\right],~0<c,~c\neq1.
$$
It is easy to see that from the previous criterion we can recover \eqref{eq:lorden} by letting $c\to1$. As in \eqref{eq:Lorden_opt} we are interested in minimizing $\hat{\cJ}_{\rm L}(T)$ over all stopping times that satisfy the same false alarm constraint $\Exp_\infty[T]\geq\gamma\geq1$. The optimum stopping time (see \cite{Poor}) has the following CUSUM-like form
$$
\hat{Y}_t=\max\{\hat{Y}_{t-1},1\}{c\ell_t},~\hat{Y}_0=0;~~\hat{T}_{\rm C}=\inf\{t>0:~\hat{Y}_t\geq\nu\}.
$$
We can then verify that $\hat{T}_{\rm C}$ is reduced to the Shewhart test when $\gamma\in[1,\Pro_\infty^{-1}(\ell_1\ge1/c)]$. If $0<c<1$, the previous interval is clearly larger than the one obtained in the classical $c=1$ case. The range of false alarm rates just specified can be quite significant if the two pdfs differ drastically, namely when we have ``large changes''. Let us demonstrate this fact with a simple example.

\textsc{Example 1}: Consider the detection of a change in the mean of a Gaussian i.i.d.~process of unit variance, from 0 to $\mu>0$. We can then see that when
$$
1\leq\gamma\leq\frac{1}{\Phi(-0.5\mu)},
$$
where $\Phi(x)$ is the cdf of a standard Gaussian, CUSUM is reduced to Shewhart with corresponding maximal average detection delay satisfying
$$
1\leq\cJ_{\rm L}(\TS)\leq\frac{1}{\Phi(0.5\mu)}.
$$
If we select $\mu$ such that $\Phi(-0.5\mu)=0.001$, resulting in $\mu=6.1805$, this allows for average false alarm periods in the interval $1\leq\gamma\leq1000$, when the corresponding detection delay is, at worst, equal to 1.001; performance which, undoubtedly, can satisfy any exigent user.

Our previous discussion corroborates what is already known in the literature, namely, that the Shewhart test behaves extremely well when changes are ``large'' while in the case of ``small'' changes one needs to resort to CUSUM. Actually, it is clear that this optimal behavior of the Shewhart test is inherited from the optimality of CUSUM.

Even though the previous result concerning the Shewhart test is interesting, it is nevertheless theoretically restricted since it covers only a limited range of false alarm rates. In the next section we will demonstrate that this simple detection rule is in fact optimum according to a number of intriguing criteria. We would also like to mention that in Section\,\ref{ssec:4.1} we will return to this optimality property of Shewhart and extend it to the case of multiple possibilities under the post-change regime. 

\section{The probability maximizing approach}
Let us now adopt the alternative performance measures introduced in Section\,\ref{ssec:1.0} and analyze the special case $m=1$. As mentioned, this corresponds to the probability of the event that detection will be achieved with the first observation under the alternative regime. Therefore we consider
\begin{align}
\ccJ_{\rm S}(T)&=\Pro\big(T=(\tau+1)^+|T>\tau\big) \label{eq:shiryaev2}\\
\ccJ_{\rm L}(T)&=\inf_{t\ge0}\esinf\Pro_t(T=t+1|\ccF_t,T>t)\label{eq:lorden2}\\
\ccJ_{\rm P}(T)&=\inf_{t\ge0}\Pro_t(T=t+1|T>t)\label{eq:pollak2},
\end{align}
corresponding to \eqref{eq:shiryae},\eqref{eq:lorde},\eqref{eq:polla} respectively, with $x^+=\max\{x,0\}$.
In the previous measures we define the value of the conditional probability to be 1 when $\{T>\tau\}$ or $\{T>t\}$ (hence also $\{T=\tau+1\}$ or $\{T=t+1\}$ respectively) is the empty set. 
Additionally, we note that in the case of Shiry\-aev's modified measure \eqref{eq:shiryaev2}, due to the existence of the prior probability, it is possible to distinguish between the events $\tau\leq-1$ and $\tau=0$. In the former case the soonest we can hope to detect the change is at time 0. This is the reason why in our criterion we use $(\tau+1)^+$ instead of $(\tau+1)$. In the other two measures this modification is unnecessary since, due to lack of prior information, a change before 0 cannot be distinguished from a change at 0.

Regarding now the stopping time $T$, we need to properly enrich the $\sigma$-algebra $\ccF_0$ so as randomization is permitted at time 0. In particular, at time 0, with probability $\varpi$ we decide to stop at 0 and take no samples and with probability $(1-\varpi)$ to employ a standard stopping time that requires sampling. Probability $\varpi$ is selected independently from the observations. This slight modification in the definition of our stopping time is absolutely necessary for Shiryaev's formulation while for the Lorden and Pollak setup it is needed only for technical reasons.

We continue our presentation by examining various optimality problems defined with the help of the previous performance measures in combination with proper false alarm constraints. We start with Shiryaev's Bayesian setup.

\subsection{Modified Shiryaev criterion}
Shiryaev \cite{Shiryaev} considered the changetime $\tau$ to be random, independent from the observations, with a zero modified exponential prior of the form\footnote{There is a slight difference between the current definition of the prior and the one encountered in the literature. This is because in our approach, $\tau$ is the last time instant under the nominal regime, whereas in the literature $\tau$ is conventionally considered as the first instant under the alternative.}: $\Pro(\tau\leq-1)=\pi$ and $\Pro(\tau=t)=(1-\pi)p(1-p)^t,~t\geq0$; where $\pi\in[0,1]$ and $p\in(0,1]$. Combining \eqref{eq:shiryaev2} with the classical constraint on the false alarm probability used in Baysian approaches, we propose the following constrained optimization problem
\begin{multline}
\sup_T\ccJ_{\rm S}(T)=\sup_T\Pro\big(T=(\tau+1)^+|T>\tau\big);\\
\text{over all}~T:~\Pro(T\leq\tau)\leq\alpha,
\label{eq:Shiryaev2}
\end{multline}
where $\alpha\in(0,1)$ is a prescribed false alarm level. The next theorem identifies the optimum scheme that solves \eqref{eq:Shiryaev2}.

\begin{theorem}\label{th:1}
The optimum detection rule that solves \eqref{eq:Shiryaev2} is defined as follows:

i)~$\alpha\geq(1-\pi)$: With probability $\varpi=1$ stop at 0 without taking any samples.

ii)~$(1-\pi)>\alpha\geq(1-\pi)\frac{(1-p)\Pro_\infty(\ell_1\geq\nu^*)}{1-(1-p)\Pro_\infty(\ell_1<\nu^*)}$ where $\nu^*=\frac{\pi}{1-\pi}\frac{1-p}{p}$: With probability
$$
\varpi=\frac{\alpha}{p(1-\pi)}[1-(1-p)\Pro_\infty(\ell_1<\nu^*)]-\frac{1-p}{p}\Pro_\infty(\ell_1\geq\nu^*),
$$
decide between stopping at 0 and using the Shewhart stopping time with constant threshold $\nu^*$.

iii)~$(1-\pi)\frac{(1-p)\Pro_\infty(\ell_1\geq\nu^*)}{1-(1-p)\Pro_\infty(\ell_1<\nu^*)}>\alpha$: The optimum is the Shewhart stopping time with constant threshold $\nu$ computed from
\begin{equation}
\Pro_\infty(\ell_1\geq\nu)=\frac{\alpha p}{(1-\pi-\alpha)(1-p)}
\label{eq:th1.iii}
\end{equation}
and with randomization probability $\varpi=0$.
\end{theorem}
\begin{proof}
The proof of Theorem\,\ref{th:1} is presented in the Appendix.
\end{proof}
\indent The exponential prior model is theoretically very appealing because it leads to well defined optimal stopping problems. However one of its key weaknesses is the need to properly specify the parameter pair $(\pi,p)$. If the two quantities are unknown and cannot be defined explicitly, a possible means to overcome this problem is to \textit{adopt a worst-case analysis} with respect to these two parameters. We should point out that this idea, detailed in the next subsection, has no equivalent in the existing literature for the classical Shiryaev criterion.

\subsection{Max-min version of the modified Shiryaev criterion}
Consider, as before, $\tau$ to be distributed according to a zero modified exponential with unknown parameters $\pi,p$. Let us denote our performance measure as
$$
\ccJ_{\rm S}(T,\pi,p)=\Pro\big(T=(\tau+1)^+|T>\tau\big),
$$
making explicit its dependence on the parameter pair $(\pi,p)$. Adopting a max-min approach, we are interested in the following constrained optimization problem
\begin{multline}
\sup_{T}\inf_{\pi,p}\ccJ_{\rm S}(T,\pi,p)=\sup_{T}\inf_{\pi,p}\Pro\big(T=(\tau+1)^+|T>\tau\big)\\
\text{over all}~T:~\Exp_\infty[T]\geq\gamma\geq1,
\label{eq:Shiryaev3}
\end{multline}
As we can see, we have replaced the false alarm probability constraint, used in the previous formulation, with a constraint on the average period between false alarms, commonly encountered in min-max approaches. The next theorem presents the optimum detection rule.

\begin{theorem}\label{th:2}
Let $\nu$ be the solution of the equation
\begin{equation}
\frac{\Pro_0(\ell_1<\nu)}{\Pro_\infty(\ell_1\geq\nu)}=\gamma,
\label{eq:th2.1}
\end{equation}
then \eqref{eq:Shiryaev3} is solved by randomizing with probability $\varpi=\Pro_0(\ell_1\geq\nu)$ between stopping at 0 and using the Shewhart stopping time with constant threshold $\nu$. The resulting stopping rule is an equalizer over all parameter pairs $(\pi,p)$; while the worst-case zero modified exponential prior is the degenerate uniform obtained by selecting $\pi(p)=\nu p/(\nu p+1-p)$ and letting $p\to0$.
\end{theorem}
\begin{proof}
The proof of Theorem\,\ref{th:2} can be found in the Appendix.
\end{proof}
It is surprising that a worst-case analysis results in an optimum stopping rule that requires non-trivial randomization at 0. This is quite uncommon in min-max approaches. It is basically due to the fact that, even though we follow a worst-case approach with respect to the two parameters, the underlying setup is still Bayesian thus accepting randomized optimum solutions, as was demonstrated in Theorem\,\ref{th:1}. Let us now continue our presentation with the max-min criteria introduced in \eqref{eq:lorden2},\eqref{eq:pollak2}.

\subsection{Modified Lorden and Pollak criterion}
We propose the following optimization problem
\begin{multline}
\sup_T\ccJ_{\rm L}(T)=\sup_{T}\inf_{t\ge0}\,\esinf\,\Pro_t(T=t+1|\ccF_t,T>t);\\~~\text{over all}~T:~\Exp_\infty[T]\geq\gamma\geq1,
\label{eq:Lorden2}
\end{multline}
where we maximize Lorden's modified measure \eqref{eq:lorden2} under the classical constraint on the average false alarm period. Similarly for Pollak's modified criterion \eqref{eq:pollak2}, we have
\begin{multline}
\sup_T\ccJ_{\rm P}(T)=\sup_{T}\inf_{t\ge0}\,\Pro_t(T=t+1|T>t);\\~~\text{over all}~T:~\Exp_\infty[T]\geq\gamma\geq1.
\label{eq:Pollak2}
\end{multline}
The following theorem offers the solution to both problems.

\begin{theorem}\label{th:3} The optimum stopping time that solves the max-min problems in \eqref{eq:Lorden2} and \eqref{eq:Pollak2} is the Shewhart test with constant threshold $\nu$ computed from the equation $\Pro_\infty(\ell_1\geq\nu)=1/\gamma$.
\end{theorem}
\begin{proof}
The proof for \eqref{eq:Pollak2} (actually under a more general semi-Bayesian setting) is given in Pollak and Krieger \cite{Pollak2}, while the one for \eqref{eq:Lorden2} is detailed in the Appendix.
\end{proof}
We note that in the case of Pollak's modified measure we have an \textit{exact} optimality result. This should be compared with the original criterion $\cJ_{\rm P}(T)$ in \eqref{eq:pollak} where (third-order) asymptotically optimum detection rules are available (see \cite{Pollak},\cite{Tartakovsky}).

The simplicity of the probability maximizing approach allows for the straightforward solution of problems which, in the classical changepoint literature (involving expected delays), have been open for many years. It is worth analyzing two such characteristic cases in detail and develop the corresponding optimal solutions.

\section{Independent, non-identically distributed observations}
Let $\{f_{\infty,t}(x)\}$, $\{f_{0,t}(x)\}$ denote two pdf sequences and consider the case where the observation process $\{\xi_t\}$ is independent but not identically distributed, following the first pdf sequence up to some changetime $\tau$ and switching to the second after $\tau$. We are interested in detecting the change optimally following the max-min approach proposed in \eqref{eq:Lorden2} or \eqref{eq:Pollak2}. We recall that the likelihood ratio $\ell_t=f_{0,t}(\xi_t)/f_{\infty,t}(\xi_t)$ has now time-varying statistics. We have the following theorem that provides the optimum solution to both problems.

\begin{theorem}\label{th:4}
The optimum stopping time that solves \eqref{eq:Lorden2} and \eqref{eq:Pollak2} for the case of independent and non-identically distributed observations, is the Shewhart stopping time $\TS=\inf\{t>0:\ell_t\geq\nu_t(\beta)\}$, where the sequence of thresholds $\{\nu_t(\beta)\}$ is obtained by solving the equations
\begin{equation}
\Pro_{0,t}\big(\ell_t\geq\nu_t(\beta)\big)=\beta;~\forall t>0,
\label{eq:th4.1}
\end{equation}
with parameter $\beta\in(0,1)$. Assuming for each $\beta$ that
\begin{equation}
\sup_{t>0}\Pro_{\infty,t}\big(\ell_t<\nu_t(\beta)\big)<1,
\label{eq:th4.1.2}
\end{equation}
this parameter is specified by requiring the false alarm constraint to be satisfied with equality, that is,
\begin{equation}
\Exp_\infty[\TS]=1+\sum_{t=1}^\infty\prod_{l=1}^t\Pro_{\infty,l}\big(\ell_l<\nu_l(\beta)\big)=\gamma.
\label{eq:th4.2}
\end{equation}
\end{theorem}
\begin{proof}
The proof of Theorem\,\ref{th:4} can be found in the Appendix.
\end{proof}
Due to the time-varying statistics, the threshold sequence needs to be time-varying as well. With \eqref{eq:th4.1} we assure that the Shewhart test is an equalizer over time, a very important property for proving its optimality. This is indeed true since $\Pro_t(\TS=t+1|\ccF_t,\TS>t)=\Pro_t(\TS=t+1|\TS>t)=\Pro_{0,t}(\ell_t\geq\nu_t(\beta))=\beta$. Of course this condition still generates an ambiguity since $\beta$ is unknown. This last parameter is then specified by forcing the Shewhart stopping time to satisfy the false alarm constraint with equality through \eqref{eq:th4.2}. Condition \eqref{eq:th4.1.2} guarantees summability of the series in \eqref{eq:th4.2} and also simplifies, considerably, the proof of our theorem. It can be relaxed but at the expense of a far more involved analysis.

\vskip0.2cm
\textsc{Example 2}: Consider the case where $f_{\infty,t}(x)$ is time invariant Gaussian with mean 0 and variance 1, while $f_{0,t}(x)$ is Gaussian with mean $\mu_t>0$ and variance 1. The sequence of thresholds then becomes
$$
\nu_t(\beta)=e^{0.5\mu_t^2+\mu_ts(\beta)};~~\text{where}~s(\beta)=\Phi^{-1}(1-\beta),
$$
and $\Phi^{-1}(x)$ denotes the inverse cdf of a standard Gaussian. Assumption \eqref{eq:th4.1.2} is valid if the sequence of means $\{\mu_t\}$ is upper bounded by a finite constant. To find $\beta$, we observe that
$\Pro_\infty(\ell_t\leq\nu_t)=\Phi(\mu_t+s(\beta))$. Since there is a one-to-one correspondence between $\beta\in(0,1)$ and $s(\beta)\in\Real$, we can instead solve \eqref{eq:th4.2} for $s$, that is,
$$
1+\sum_{t=1}^\infty\prod_{l=1}^t\Phi(\mu_l+s)=\gamma
$$
and compute the optimum performance as $\beta=1-\Phi(s)=\Phi(-s)$.

\section{Multiple post-change probability measures}
Consider now the change detection problem with more than one post-change possibilities. Our observation sequence $\{\xi_t\}$ is i.i.d.~before and after the change with a common
pdf $f_\infty(\xi)$ before the change and two\footnote{Extension to more than two pdfs is straightforward.} different pdf possibilities $f_0^1(\xi),f_0^2(\xi)$ after the change. Following a pure non-Bayesian approach (see Pollak and Krieger \cite{Pollak2} for semi-Bayesian formulations) we extend the definition of our performance measures in order to account for the multiple post-change distributions. Define
\begin{align*}
\ccJ_{\rm L}(T)&=\min_{i=1,2}\inf_{t\ge0}{\esinf}\;\Pro_t^i(t<T\le t+m|\ccF_t,T>t)\\
\ccJ_{\rm P}(T)&=\min_{i=1,2}\inf_{t\ge0}\Pro_t^i(t<T\le t+m|T>t),
\end{align*}
where $\Pro_t^i$ is the measure induced by a change at time $t$ with the alternative pdf being $f_0^i(\xi)$. Consequently in our criterion we include an additional minimization over the possible alternative measures.

Limiting, again, ourselves to the special case $m=1$, we are interested in solving the following constrained optimization problems
\begin{multline}
\sup_T\ccJ_{\rm L}(T)=\sup_T\min_{i=1,2}\inf_{t\ge0}\esinf\Pro_t^i(T=t+1|\ccF_t,T>t),\\
\text{over all}~T:~\Exp_\infty[T]\ge\gamma\geq1,
\label{eq:lorden4}
\end{multline}
for the Lorden and
\begin{multline}
\sup_T\ccJ_{\rm P}(T)=\sup_T\min_{i=1,2}\inf_{t\ge0}\Pro_t^i(T=t+1|T>t),\\
\text{over all}~T:~\Exp_\infty[T]\ge\gamma\geq1,
\label{eq:pollak4}
\end{multline}
for the Pollak criterion. We note that we have two sequences of likelihood ratios, namely $\{\ell_t^1\}$ and $\{\ell_t^2\}$ defined as $\ell_t^i=f_0^i(\xi_t)/f_\infty(\xi_t),~i=1,2$. For each $q\in[0,1]$ we define a threshold $\nu(q)\geq0$, so that the following version of the Shewhart test
\begin{equation}
\TS(q)=\inf\{t>0:(1-q)\ell_t^1+q\ell_t^2\geq\nu(q)\},
\label{eq:Sq}
\end{equation}
satisfies the equation
\begin{equation}
\Pro_\infty\big((1-q)\ell_1^1+q\ell_1^2\geq\nu(q)\big)=\frac{1}{\gamma}.
\label{eq:nuq}
\end{equation}
The next theorem demonstrates that by proper selection of the parameter $q$, the corresponding stopping time solves both optimization problems.

\begin{theorem}\label{th:5}
For the solution of \eqref{eq:lorden4} and \eqref{eq:pollak4} we distinguish three cases:

i) If\/ $\Pro_0^2(\ell_1^1\geq\nu(0))\geq\Pro_0^1(\ell_1^1\geq\nu(0))$, then the optimum test is $\TS(0)$.

ii) If\/ $\Pro_0^1(\ell_1^2\geq\nu(1))\geq\Pro_0^2(\ell_1^2\geq\nu(1))$ then the optimum test is $\TS(1)$.

iii) If there is $q\in(0,1)$ such that 
\begin{equation}
\Pro_0^1\big((1-q)\ell_1^1+q\ell_1^2\geq\nu(q)\big)=\Pro_0^2\big((1-q)\ell_1^1+q\ell_1^2\geq\nu(q)\big),
\label{eq:th5.1}
\end{equation}
then the optimum test is $\TS(q)$. For each $\gamma\geq1$, only one of i), ii) and iii) applies.
\end{theorem}

\begin{proof}
The proof of Theorem\,\ref{th:5} can be found in the Appendix.
\end{proof}

We can use the previous outcome to find solutions for Lorden's \textit{original} criterion involving average detection delays when there are multiple post-change probabilities. The goal is to obtain a result similar to the one presented in Section\,\ref{ssec:1.1} for the Shewhart rule of Theorem\,\ref{th:5}.

\subsection{Multiple post-change measures with Lorden's original criterion}\label{ssec:4.1}
Consider the Lorden criterion in \eqref{eq:lorden} properly extended to cover multiple post-change probability distributions. In particular we propose
$$
\cJ_{\rm L}(T)=\max_{i=1,2}\sup_{t\ge0}\,\esup\,\Exp^i_t[T-t|\ccF_t,T>t].
$$
We are then interested in the following min-max constrained optimization problem
\begin{multline}
\inf_T\cJ_{\rm L}(T)=\inf_T\max_{i=1,2}\sup_{t\ge0}\,\esup\,\Exp^i_t[T-t|\ccF_t,T>t]\\
\text{over all}~T:~\Exp_\infty[T]\geq\gamma\geq1.
\label{eq:lorden_M}
\end{multline}

This problem has been open for many years. Existing results typically refer to the two-sided CUSUM (2-CUSUM) and demonstrate that this rule exhibits different levels of asymptotic optimality. For example in Hadjiliadis and Moustakides \cite{Hadjiliadis1} and Hadjiliadis and Poor \cite{Hadjiliadis2}, it is proved that specially designed 2-CUSUM tests enjoy second and third order asymptotic optimality when detecting changes in the constant drift of a Brownian Motion. Dragalin \cite{Dragalin} provides first order asymptotically optimum 2-CUSUM rules for the case of single parameter exponential families. 

With the next theorem we present the analog of Section\,\ref{ssec:1.1} for the case of two post-change probability measures. In particular we demonstrate that the Shewhart test of Theorem\,\ref{th:5} can be the \textit{exact} solution to \eqref{eq:lorden_M} provided threshold $\nu$ (hence parameter $\gamma$) takes values within a range that we explicitly identify. The next theorem presents the precise form of our claim. We recall that the two likelihood ratios $\ell_t^i$ are known functions of the observation $\xi_t$.

\begin{theorem}\label{th:6}
With $\TS(q)$ defined in \eqref{eq:Sq} and \eqref{eq:nuq}, we distinguish three cases that can provide partial solution to \eqref{eq:lorden_M}:

i) If\/ $\Pro_0^2\big(\ell_1^1\geq\nu(0)\big)\geq\Pro_0^1\big(\ell_1^1\geq\nu(0)\big)$ with $1\geq\nu(0)\geq0$, then the optimum test is $\TS(0)$.

ii) If\/ $\Pro_0^1\big(\ell_1^2\geq\nu(1)\big)\geq\Pro_0^2\big(\ell_1^2\geq\nu(1)\big)$  with $1\geq\nu(1)\geq0$, then the optimum test is $\TS(1)$.

iii) If there is $q\in(0,1)$ with $1\geq\nu(q)\geq0$ such that 
\begin{equation}
\Pro_0^1\big((1-q)\ell_1^1+q\ell_1^2\geq\nu(q)\big)=\Pro_0^2\big((1-q)\ell_1^1+q\ell_1^2\geq\nu(q)\big),
\label{eq:th6.1}
\end{equation}
and 
\begin{multline}
\min\{q+(1-q)\inf_{\xi_1\in\cA_1\cap\cA_2^c}\ell_1^1,(1-q)+q\inf_{\xi_1\in\cA_1^c\cap\cA_2}\ell_1^2\}\geq\\
\nu(q)\geq
\inf_{\xi_1\in\cA_1\cap\cA_2}\{(1-q)\ell_1^1+q\ell_1^2\}
\label{eq:th6.2}
\end{multline}
where $\cA_i=\{\xi_1:\ell_1^i\leq1\}$ and $\cA_i^c$ its complement, then the optimum test is $\TS(q)$.
\end{theorem}

\begin{proof}
The proof of Theorem\,\ref{th:5} can be found in the Appendix.
\end{proof}

Even though the extent of this result is clearly limited, it is nontheless the first time we have a nonasymptotic solution for Lorden's formulation when there are multiple distributions under the alternative regime. Theorem\,\ref{th:6} also establishes that 2-CUSUM is \textit{not} strictly optimum (at least not in the sense of \eqref{eq:lorden_M}) despite its very strong asymptotic optimality properties. Finally we need to mention that it is not possible to recover the same non-asymptotic result by assigning specific prior probabilities to the post-change measures (i.e. following the semi-Baysian idea of \cite{Pollak2}). The extra freedom enjoyed by considering each probability measure separately is critical in demonstrating the optimality of the Shewhart test in the sense of Lorden.

\vskip0.2cm
\textsc{Example 3:} Consider the Gaussian case where under the nominal regime the samples are i.i.d. with mean 0 and variance 1 whereas under the alternative they can have two possible means $\pm\mu,~\mu>0$ with unit variance. Let us apply case~iii) of Theorem\,\ref{th:6}. Due to symmetry it is sufficient to select $q=0.5$ to satisfy \eqref{eq:th6.1}.

The two likelihood ratios $\ell_1^i$ as functions of the observation $\xi_1$ are equal to $e^{-0.5\mu^2\pm\mu\xi_1}$ and the sets of interest are $\cA_1=\{\xi_1:\xi_1\leq0.5\mu\}$ and $\cA_2=\{\xi_1:\xi_1\geq-0.5\mu\}$. We can now compute the critical range for threshold $\nu$ from \eqref{eq:th6.2}. Since $\inf_{\xi_1<-0.5\mu}\ell_1^1=\inf_{\xi_1>0.5\mu}\ell_1^2=0$ and $\inf_{-0.5\mu\leq\xi_1\leq0.5\mu}0.5(\ell_1^1+\ell_1^2)=e^{-0.5\mu^2}$, we have $0.5\geq\nu\geq e^{-0.5\mu^2}$. This interval is nonempty when $\mu>\sqrt{2\log2}=1.1774$ and gives rise to the following range for $\gamma$
$$
1\leq\gamma\leq\frac{1}{2\Phi(-0.5\mu+\delta)},~\text{where}~\delta=-\frac{1}{\mu}\log\left(\frac{1+\sqrt{1-4e^{-\mu^2}}}{2}\right),
$$
with the worst-case average detection delay satisfying
$$
1\leq\cJ_{\rm L}(\TS)\leq\frac{1}{\Phi(0.5\mu+\delta)+\Phi(-1.5\mu+\delta)}.
$$
Using the same numerical value we adopted in Example\,1 for the one-sided case, namely, $\mu=6.1805$, we obtain $1\leq\gamma\leq500$ while the optimum detection delay becomes, at worst, 1.001. Compared to Example\,1, as we can see, the range of $\gamma$ where the Shewhart test is optimum is reduced to half.

\vskip0.2cm
\textsc{Remark:} Because with the maximizing probability approach we focus on a single sample after the change, it turns out that Shewhart is optimum for \textit{transient changes} as well. Specifically, the same proofs go through for any type of change provided \textit{it lasts at least one sample} (which is necessary for a change to exist). Clearly this is an additional distinct optimality characteristic enjoyed by this simple detection rule. As we know, the Shiryaev, CUSUM and Shiryaev-Roberts tests lose their optimality if the change does not last indefinitely after its occurrence.


\appendix
\section*{Appendix: Proofs}
\textsc{Proof of Theorem \ref{th:1}.}
We begin our analysis by writing the performance measure in a more detailed form. We have
\begin{equation}
\ccJ_{\rm S}(T)=\Pro(T=(\tau+1)^+|T>\tau)=\frac{\Pro(T=(\tau+1)^+)}{\Pro(T>\tau)}.
\end{equation}
Since $T\geq0$, for the numerator we can write
\begin{multline}
\Pro(T=(\tau+1)^+)=\Pro(\tau\leq-1)\Pro(T=0)+\sum_{t=0}^\infty\Pro(\tau=t)\Pro_t(T=t+1)\\
=\pi\varpi+(1-\pi)p\sum_{t=0}^\infty(1-p)^{t}\Pro_t(T=t+1)\\
=\pi\varpi+(1-\pi)p\sum_{t=0}^\infty(1-p)^{t}\Exp_\infty[\ell_{t+1}\ind{T=t+1}]\\
=\pi\varpi+\frac{(1-\pi)p}{(1-p)}\Exp_\infty[(1-p)^T\ell_T\ind{T>0}]\\
=\pi\varpi+\frac{(1-\pi)p}{(1-p)}\Exp_\infty[(1-p)^T\ell_T|T>0]\Pro(T>0)\\
=\pi\varpi+\frac{(1-\pi)p}{(1-p)}\Exp_\infty[(1-p)^T\ell_T|T>0](1-\varpi).
\label{eq:Sh_num}
\end{multline}
Similarly for the denominator, since $\{T>t\}\in\ccF_t$ and $T\geq0$, we have
\begin{multline}
\Pro(T>\tau)=\Pro(\tau\leq-1)\Pro(T>-1)+\sum_{t=0}^\infty\Pro(\tau=t)\Pro_\infty(T>t)\\
=\pi+(1-\pi)p\sum_{t=0}^\infty(1-p)^{t}\Pro_\infty(T>t)
=\pi+(1-\pi)p\Exp_\infty\left[\sum_{t=0}^{T-1}(1-p)^t\right]\\
=\pi+(1-\pi)\Exp_\infty[1-(1-p)^T]=\pi+(1-\pi)\Exp_\infty\left[\big(1-(1-p)^T\big)\ind{T>0}\right]\\
=\pi+(1-\pi)\Exp_\infty[1-(1-p)^T|T>0]\Pro(T>0)\\
=\pi+(1-\pi)\{1-\Exp_\infty[(1-p)^{T}|T>0]\}(1-\varpi),
\label{eq:Sh_den}
\end{multline}
with the third last equality being true because $1-(1-p)^0=0$.
Combining \eqref{eq:Sh_num} and \eqref{eq:Sh_den} we have the following form for the modified Shiryaev measure
\begin{equation}
\ccJ_{\rm S}(T)=\frac{\pi\varpi+\frac{(1-\pi)p}{1-p}\Exp_\infty[(1-p)^T\ell_T|T>0](1-\varpi)}{\pi+(1-\pi)\{1-\Exp_\infty[(1-p)^T|T>0]\}(1-\varpi)}.
\label{eq:sh_analytic}
\end{equation}

Next, we distinguish different possibilities depending on the value of $\alpha$. For case~i) where $\alpha\geq1-\pi$ by selecting $\varpi=1$, in other words stopping at 0 with probability 1, as we can see from \eqref{eq:sh_analytic}, yields $\ccJ_{\rm S}(T)=1$ which is the maximum possible value for our criterion (since it is a probability). On the other hand the denominator, which is the complement of the false alarm probability, from \eqref{eq:Sh_den} is equal to $\pi$. This means that the false alarm probability is $1-\pi$ thus satisfying the constraint.

Let now $1-\pi>\alpha$. Since the stopping time $T$ must satisfy the false alarm constraint, this suggests that the denominator, by being the complement of the false alarm probability, is no smaller than $1-\alpha$. We are going to show that in order to maximize the performance measure it is sufficient to limit ourselves to stopping times that satisfy the false alarm constraint with equality. This equality will be achieved by modifying the randomization probability $\varpi$ in a way that will improve (increase) the value of the criterion $\ccJ_{\rm S}(T)$.

As we can see from \eqref{eq:sh_analytic} both, the numerator and the denominator are linear functions of $\varpi$ and the ratio takes its maximal value (equal to 1) for $\varpi=1$. We can therefore conclude that the ratio is an increasing function of $\varpi$. If $T$ is such that the denominator is strictly greater than $1-\alpha$ and since we are in the case where $\pi<1-\alpha$, this suggests that, necessarily, we have $1-(1-\pi)\Exp_\infty[(1-p)^T|T>0]>1-\alpha$. Consequently by replacing $\varpi$ with a \textit{larger} value $\varpi'>\varpi$ we can make the denominator exactly equal to $1-\alpha$. Making the same change in the numerator, due to the monotonicity with respect to $\varpi$ this will result in an overall increase of our performance measure. Therefore, without loss of generality, we may limit ourselves to stopping times that satisfy the false alarm constraint with equality.

The previous observation suggests that we can maximize the numerator in \eqref{eq:sh_analytic} subject to the constraint that the denominator is equal to $(1-\alpha)$. Using the Lagrange multiplier technique we define the following criterion $\ccG(T)$ that combines the numerator and the constraint
\begin{align*}
\ccG(T)&=\pi\varpi+\frac{(1-\pi)p}{1-p}\Exp_\infty[(1-p)^{T}\ell_{T}|T>0](1-\varpi)\\
&~~~+\lambda\left\{\pi\varpi+\{1-(1-\pi)\Exp_\infty[(1-p)^{T}|T>0]\}(1-\varpi)\right\}\\
&=\pi(1+\lambda)\varpi\\
&~~~+\left\{\lambda+(1-\pi)\Exp_\infty\left[(1-p)^{T}\left(\frac{p}{1-p}\ell_{T}-\lambda\right)|T>0\right]\right\}(1-\varpi),
\end{align*}
with $\lambda$ being the corresponding Lagrange multiplier.
The goal, now, is first to maximize $\ccG(T)$ over $T>0$ and then over the randomization probability $\varpi\in[0,1]$. Fixing $\varpi$ and maximizing over $T>0$ means that we need to maximize the expression
$$
\hat{\ccG}(T)=\Exp_\infty\left[(1-p)^{T}\left(\frac{p}{1-p}\ell_{T}-\lambda\right)|T>0\right].
$$
For $\nu\geq0$ consider the following specific value of the Lagrange multiplier
$$
\lambda=\frac{\nu}{1-p}\{1-(1-p)\Pro_\infty(\ell_1<\nu)\}-\Pro_0(\ell_1\geq\nu).
$$
Using standard optimal stopping theory it is then straightforward to show that $\hat{\ccG}(T)$ is 
maximized by the Shewhart stopping time defined in \eqref{eq:Shewhart} with constant threshold $\nu$. The corresponding optimum performance can then be computed as follows
\begin{multline*}
\hat{\ccG}(\TS)=\sum_{t=1}^\infty(1-p)^t[\Pro_\infty(\ell_1<\nu)]^{t-1}\left\{\frac{p}{1-p}\Pro_0(\ell_1\geq\nu)-\lambda\Pro_\infty(\ell_1\geq\nu)\right\}\\
=\frac{p\Pro_0(\ell_1\geq\nu)-\lambda(1-p)\Pro_\infty(\ell_1\geq\nu)}{1-(1-p)\Pro_\infty(\ell_1<\nu)}=
\frac{p}{1-p}\nu-\lambda.
\end{multline*}
The last equality can be verified by directly substituting the definition of the Lagrange multiplier $\lambda$.
Using this result in the original measure, we end up with the following inequality
\begin{multline}
\ccG(T)\leq\pi(1+\lambda)\varpi+\left[\pi(1+\lambda)+(1-\pi)\frac{p}{1-p}\nu-\pi\right](1-\varpi).
\label{eq:Gopt}
\end{multline}
From \eqref{eq:Sh_den}, we can also compute the corresponding false alarm probability which must be set equal to $\alpha$ (we must satisfy the constraint with equality)
\begin{equation}
(1-\pi)\left\{
\varpi+\frac{(1-p)\Pro_\infty(\ell_1\geq\nu)}{1-(1-p)\Pro_\infty(\ell_1<\nu)}(1-\varpi)
\right\}=\alpha.
\label{eq:FAopt}
\end{equation}

We are now left with the definition of the randomization probability $\varpi$. Selecting $\varpi$ optimally amounts to maximizing the right hand side in \eqref{eq:Gopt} over $\varpi$. We observe that the corresponding expression is a convex combination of the value $\pi(1+\lambda)$, which is the gain obtained when stopping at 0, and $\pi(1+\lambda)+(1-\pi)\frac{p}{1-p}\nu-\pi$, which is the gain resulting by employing $\TS$ for $t>0$. Clearly we are going to put all the probability mass on the largest gain. Consequently, when $\nu>\nu^*=\frac{\pi}{1-\pi}\frac{1-p}{p}$ the gain provided by $\TS$ exceeds the gain obtained by stopping at 0, therefore in this case we select $\varpi=0$. Of course $\nu$ must be such that the Shewhart test satisfies the false alarm constraint with equality. From \eqref{eq:FAopt} by substituting $\varpi=0$ we can see that the constraint is satisfied when $\nu$ is computed through equation \eqref{eq:th1.iii}. This equation has always a solution that exceeds $\nu^*$ as long as $\alpha$ takes values in the interval specified in case iii). When $\nu=\nu^*$, stopping at 0 provides exactly the same gain as the Shewhart test $\TS$ with threshold $\nu^*$. Therefore we can randomize between the two possibilities with any probability $\varpi$. However, since we need to satisfy the false alarm constraint with equality, from \eqref{eq:FAopt} with $\nu=\nu^*$ we can solve for $\varpi$, and obtain the optimum $\varpi$ depicted in case~ii). The resulting value corresponds to a legitimate probability $\varpi\in[0,1]$ when $\alpha$ is within the limits prescribed for this case. This concludes our proof.\qed

\vskip0.2cm
\textsc{Proof of Theorem \ref{th:2}.} The proof will rely on the analysis we applied in the proof of Theorem\,\ref{th:1}. In order to solve the max-min problem defined in our theorem it is sufficient to show the existence of a tripple $(T^*,\pi^*,p^*)$ such that the following saddle-point relation holds
\begin{equation}
\ccJ_{\rm S}(T,\pi^*,p^*)\leq \ccJ_{\rm S}(T^*,\pi^*,p^*)\leq
\ccJ_{\rm S}(T^*,\pi,p),
\label{eq:AA.1}
\end{equation}
for all stopping times $T$ that satisfy the false alarm constraint. It is well known that whenever a saddle-point solution exists it is also max-min optimum. Indeed note that if $T^*$ satisfies \eqref{eq:AA.1} then we can write
$$
\inf_{\pi,p}\ccJ_{\rm S}(T,\pi,p)\leq\ccJ_{\rm S}(T,\pi^*,p^*)\leq\ccJ_{\rm S}(T^*,\pi^*,p^*)=\inf_{\pi,p}\ccJ_{\rm S}(T^*,\pi,p),
$$
where the first inequality is obvious; the second corresponds to the left hand side inequality in \eqref{eq:AA.1} and the last equality is equivalent to the right hand side inequality in \eqref{eq:AA.1}. Consequently $T^*$ solves the max-min problem and the parameter pair $(\pi^*,p^*)$ corresponds to the worst-case (least-favorable) exponential prior.

To show \eqref{eq:AA.1}, let us first define our candidate optimum stopping time $T^*$. Consider \eqref{eq:th2.1}, and observe that for $\nu\to0$ the left hand side tends to 0, whereas for $\nu\to\infty$ the same expression tends to $\infty$. Furthermore the ratio is a strictly increasing and continuous function of $\nu$ (see Footnote \ref{foot:3}). Due to this continuity and strict monotonicity the equation has a unique solution $\nu$. With the help of this threshold value our candidate detection rule $T^*$ consists in randomizing with probability $\varpi^*=\Pro_0(\ell_1\geq\nu)$ between stopping at 0 and using the Shewhart test $\TS$ with constant threshold $\nu$. For $T^*$ we observe that $\Exp_\infty[T^*]=(1-\varpi^*)\Exp_\infty[\TS]=\Pro_0(\ell_1<\nu)/\Pro_\infty(\ell_1\geq\nu)=\gamma$, suggesting that it satisfies the false alarm constraint with equality. 

We first demonstrate that $T^*$ satisfies the right hand side in \eqref{eq:AA.1}. For any parameter pair $(\pi,p)$, after recalling that on $\{T^*>0\}$ we have $T^*=\TS$, we can verify using \eqref{eq:sh_analytic}, that
$$
\ccJ_{\rm S}(T^*,\pi,p)=\frac{\pi\varpi^*+\frac{(1-\pi)p}{1-p}\Exp_\infty[(1-p)^{\TS}\ell_{\TS}](1-\varpi^*)}{\pi\varpi^*+\{1-(1-\pi)\Exp_\infty[(1-p)^{\TS}]\}(1-\varpi^*)}=\Pro_0(\ell_1\geq\nu).
$$
The last equality is true since we can immediately compute $\Exp_\infty[(1-p)^{{\TS}}\ell_{{\TS}}]=(1-p)\Pro_0(\ell_1\geq\nu)/\{1-(1-p)\Pro_\infty(\ell_1<\nu)\}$ and $\Exp_\infty[(1-p)^{{\TS}}]=(1-p)\Pro_\infty(\ell_1\geq\nu)/\{1-(1-p)\Pro_\infty(\ell_1<\nu)\}$.
As we realize, the resulting performance of $T^*$ is independent from $(\pi,p)$ therefore the stopping rule is an equalizer with respect to the two parameters. This, in turn, suggests that the right hand side in \eqref{eq:AA.1} is trivially satisfied with equality.

Showing the left hand side inequality requires more work. Note that we need to define the worst-case parameter pair $(\pi^*,p^*)$. Unfortunately this pair turns out to be a limiting case corresponding to an exponential prior that tends to a degenerate uniform. More specifically, for $p>0$ we solve for $\pi$ the following equation
$$
\frac{\pi}{1-\pi}\frac{1-p}{p}=\nu,
$$
resulting in $\pi(p)=\nu p/(1-p+\nu p)$. The parameter pair $(\pi(p),p)$ with $p\to0$ yields the worst-case exponential prior we are interested in. Consequently for the left hand side inequality we need to prove that
\begin{equation}
\lim_{p\to0}\ccJ_{\rm S}(T,\pi(p),p)\leq\lim_{p\to0}\ccJ_{\rm S}(T^*,\pi(p),p)=\Pro_0(\ell_1\geq\nu),
\label{eq:left_in}
\end{equation}
over all $T$ satisfying the false alarm constraint $\Exp_\infty[T]=(1-\varpi)\Exp_\infty[T|T>0]\geq\gamma$. 

Fix a sufficiently small $\epsilon>0$ so that $\varpi^*_\epsilon=\Pro_0(\ell_1\geq\nu)+\epsilon<1$. Define the false alarm level 
$$
\alpha_\epsilon(p)=\big(1-\pi(p)\big)\frac{p\varpi^*_\epsilon+(1-p)\Pro_\infty(\ell_1\geq\nu)}{p+(1-p)\Pro_\infty(\ell_1\geq\nu)}
$$
and the class of stopping rules
\begin{equation}
\ccA_\epsilon(p)=\{T:~\Pro(T\leq\tau)\leq\alpha_\epsilon(p)\}.
\label{eq:classA}
\end{equation}
It is then straightforward to verify that $\alpha_\epsilon(p)\in(0,1)$ and that for any probability $p$, the quantities $\pi(p),p,\alpha_\epsilon(p)$ are such that case~ii) of Theorem\,\ref{th:1} applies. This suggests that when $\ccJ_{\rm S}(T,\pi(p),p)$ is maximized over the class $\ccA_\epsilon(p)$, the optimum stopping time is to randomize between stopping at 0 and the Shewhart test $\TS$ with threshold $\nu$ using the randomization probability $\varpi^*_\epsilon$. The latter is a direct consequence of the specific definition of $\pi(p)$ and $\alpha_\epsilon(p)$. Call the resulting optimal stopping time $T^*_\epsilon$. Note also that this optimality property is true for all $1>p>0$. Using the definitions of $\pi(p)$, $\alpha_\epsilon(p)$ and \eqref{eq:Sh_den}, we can verify that the class $\ccA_\epsilon(p)$ in \eqref{eq:classA} can be equivalently written as
\begin{equation}
\ccA_\epsilon(p)=\left\{T:~(1-\varpi)\Exp_\infty\left[\frac{1-(1-p)^{T}}{p}|T>0\right]\geq\frac{\gamma-\frac{\epsilon}{\Pro_\infty(\ell_1\geq\nu)}}{1+p\gamma}\right\},
\label{eq:classA2}
\end{equation}
where we also used \eqref{eq:th2.1}.

Fix a $T$ that satisfies the false alarm constraint $\Exp_\infty[T]=(1-\varpi)\Exp_\infty[T|T>0]\geq\gamma$. As we argued before, our goal is to prove \eqref{eq:left_in}. From monotone convergence we have
$$
\lim_{p\to0}(1-\varpi)\Exp_\infty\left[\frac{1-(1-p)^{T}}{p}|T>0\right]=(1-\varpi)\Exp_\infty[T|T>0]\geq\gamma.
$$
Consequently, for any $p\in(0, p_\epsilon]$, where $p_\epsilon$ sufficiently small, we can write
$$
(1-\varpi)\Exp_\infty\left[\frac{1-(1-p)^{T}}{p}|T>0\right]\geq\gamma-\frac{\epsilon}{\Pro_\infty(\ell_1\geq\nu)}.
$$
The previous inequality, comparing with \eqref{eq:classA2}, suggests that $T\in\ccA_\epsilon(p)$ for all $0<p\leq p_\epsilon$. A direct consequence of this fact is that $\ccJ_{\rm S}(T,\pi(p),p)\leq \ccJ_{\rm S}(T^*_\epsilon,\pi(p),p)$ for all $0<p\leq p_\epsilon$. Taking the limit as $p\to0$ and using monotone convergence, we obtain
\begin{multline}
\lim_{p\to0}\ccJ_{\rm S}(T,\pi(p),p)\leq
\lim_{p\to0}\ccJ_{\rm S}(T^*_\epsilon,\pi(p),p)\\
=\frac{\varpi^*_\epsilon\nu+(1-\varpi^*_\epsilon)\Exp_\infty[\ell_{\TS}]}{\nu+(1-\varpi^*_\epsilon)\Exp_\infty[\TS]},
\label{eq:B.4}
\end{multline}
Since $\Exp_\infty[\TS]=1/\Pro_\infty(\ell_1\geq\nu)$, $\Exp_\infty[\ell_{\TS}]=\Pro_0(\ell_1\geq\nu)/\Pro_\infty(\ell_1\geq\nu)$ and $\varpi^*_\epsilon=\Pro_0(\ell_1\geq\nu)+\epsilon$, we conclude that $\Exp_\infty[\ell_{\TS}]\leq\varpi^*_\epsilon\Exp_\infty[\TS]$. Substituting in \eqref{eq:B.4} yields
$$
\lim_{p\to0}\ccJ_{\rm S}(T,\pi(p),p)\leq\varpi^*_\epsilon=\Pro_0(\ell_1\geq\nu)+\epsilon.
$$
Because this inequality is true for any sufficiently small $\epsilon>0$, we have validity of \eqref{eq:left_in}. This concludes our proof.\qed

\vskip0.2cm
\textsc{Proof of Theorem \ref{th:3}.}
If $T$ is such that $\Exp_\infty[T]=\infty$, then we can define a sufficiently large integer $M$ so that $\infty>\Exp_\infty[T_M]\geq\gamma$ where $T_M=\min\{T,M\}$. Since for $t<M$ we have $\{T=t+1\}=\{T_M=t+1\}$ and $\{T>t\}=\{T_M>t\}$, we conclude
$\Pro_t(T_M=t+1|\ccF_t,T_M>t)=\Pro_t(T=t+1|\ccF_t,T>t)$. On the other hand for $t\geq M$ it is true that $\{T_M=t+1\}=\{T_M>t\}=\varnothing$, suggesting that $\Pro_t(T_M=t+1|\ccF_t,T_M>t)=1\geq\Pro_t(T=t+1|\ccF_t,T>t)$. This means that $\ccJ_{\rm L}(T)\leq\ccJ_{\rm L}(T_M)$. The last inequality implies that we can limit ourselves to stopping times $T$ that satisfy $\infty>\Exp_\infty[T]\geq\gamma$.

From Lorden's modified measure \eqref{eq:Lorden2} we conclude that for all $t\geq0$ we can write
$$
\ccJ_{\rm L}(T)\leq\Pro_t(T=t+1|\ccF_t,T>t).
$$
Multiplying both sides with $\ind{T>t}$ and taking expectation with respect to the nominal measure yields
\begin{equation}
\ccJ_{\rm L}(T)\Pro_\infty(T>t)\leq\Pro_t(T=t+1)=\Exp_\infty[\ell_{t+1}\ind{T=t+1}].
\label{eq:C.0}
\end{equation}
Summing over all $t\geq0$ we obtain
$$
\ccJ_{\rm L}(T)\Exp_\infty[T]\leq\Exp_\infty[\ell_T],
$$
where we define $\ell_0=0$. From the previous inequality we conclude
$$
\ccJ_{\rm L}(T)\leq\frac{\Exp_\infty[\ell_T]}{\Exp_\infty[T]}=\frac{(1-\varpi)\Exp_\infty[\ell_T|T>0]}{(1-\varpi)\Exp_\infty[T|T>0]}=\frac{\Exp_\infty[\ell_T|T>0]}{\Exp_\infty[T|T>0]}.
$$

Let us examine the ratio $\Exp_\infty[\ell_T]/\Exp_\infty[T]$ over all $T$ that satisfy the constraint. Note that when $\infty>\Exp_\infty[T]=(1-\varpi)\Exp_\infty[T|T>0]>\gamma$ we can replace $\varpi$ with a larger value $\varpi'$ so that $(1-\varpi')\Exp_\infty[T|T>0]=\gamma$ without changing the value of the ratio (since it does not depend on $\varpi$). This in turn suggests that any value attained by this ratio can also be achieved by a stopping time that satisfies the constraint with equality. Using this observation we can write
\begin{equation}
\sup_{T:\Exp_\infty[T]\geq\gamma}\ccJ_{\rm L}(T)\leq \sup_{T:\Exp_\infty[T]=\gamma}\frac{\Exp_\infty[\ell_T]}{\Exp_\infty[T]}=\gamma^{-1}\sup_{T:\Exp_\infty[T]=\gamma}\Exp_\infty[\ell_T].
\label{eq:C.1}
\end{equation}

To maximize $\Exp_\infty[\ell_T]$ over all stopping times that satisfy the constraint with equality, we reduce the optimization problem into an unconstraint one using the Lagrange multiplier technique. In particular we  consider the maximization of
$$
\ccG(T)=\Exp_\infty[\ell_T-\lambda T]=(1-\varpi)\Exp_\infty[\ell_T-\lambda T|T>0].
$$
To find the optimum stopping time we will first optimize over $T>0$ and then identify the optimum randomization probability $\varpi$. Let $\nu\geq0$ be the solution of the equation $\Pro_\infty(\ell_1\geq\nu)=1/\gamma$. Define $\lambda=\Pro_0(\ell_1\geq\nu)-\nu\Pro_\infty(\ell_1\geq\nu)$. Using standard optimal stopping theory we can then conclude that $\ccG(T)$ for $T>0$ is optimized by the Shewhart test with threshold $\nu$. Since
$\Exp_\infty[\TS]=1/\Pro_\infty(\ell_1\geq\nu)=\gamma$ and $\Exp_\infty[\ell_{\TS}]=\Pro_0(\ell_1\geq\nu)/\Pro_\infty(\ell_1\geq\nu)=\gamma\Pro_0(\ell_1\geq\nu)$, if we also use the definition of $\lambda$ we conclude that
$$
\ccG(T)=(1-\varpi)\Exp_\infty[\ell_T-\lambda T|T>0]\leq(1-\varpi)\Exp_\infty[\ell_{\TS}-\lambda \TS]=(1-\varpi)\nu\leq\nu.
$$
The last inequality suggests that the optimum randomization probability is $\varpi=0$.
From the previous result we have that for any $T$ satisfying the false alarm constraint with equality, we can write
$$
\Exp_\infty[\ell_T]-\lambda \gamma=\Exp_\infty[\ell_T-\lambda T]\leq\Exp_\infty[\ell_{\TS}-\lambda \TS]=
\Exp_\infty[\ell_{\TS}]-\lambda\gamma,
$$
which implies $\Exp_\infty[\ell_T]\leq\Exp_\infty[\ell_{\TS}]=\gamma\Pro_0(\ell_1\geq\nu)$.
Observing also that for every $t\geq0$ we have $\Pro_t(\TS=t+1|\ccF_t,\TS>t)=\Pro_0(\ell_1\geq\nu)$, this means that Shewhart is an equalizer, consequently $\ccJ_{\rm L}(\TS)=\Pro_0(\ell_1\geq\nu)$. Using these two facts in \eqref{eq:C.1} leads to
\begin{multline*}
\ccJ_{\rm L}(\TS)\leq\sup_{T:\Exp_\infty[T]\geq\gamma}\ccJ_{\rm L}(T)\leq\gamma^{-1}\sup_{T:\Exp_\infty[T]=\gamma}\Exp_\infty[\ell_T]\\
\leq\gamma^{-1}\{\gamma\Pro_0(\ell_1\geq\nu)\}=\Pro_0(\ell_1\geq\nu)=\ccJ_{\rm L}(\TS),
\end{multline*}
which proves optimality for $\TS$ and concludes the proof. Exactly the same analysis applies to \eqref{eq:Pollak2}. In fact, we can simply start the proof from \eqref{eq:C.0}, which is immediately satisfied by Pollak's modified measure.\qed

\vskip0.2cm
\textsc{Proof of Theorem \ref{th:4}.} Let $\beta$ and $\{\nu_t(\beta)\}$ be such that \eqref{eq:th4.1},\eqref{eq:th4.1.2},\eqref{eq:th4.2} are satisfied.  If we define $\rho(\beta)=\sup_{t>0}\Pro_{\infty,t}\big(\ell_t<\nu_t(\beta)\big)$, then assumption \eqref{eq:th4.1.2} is equivalent to
\begin{equation}
0\leq\rho(\beta)<1.
\label{eq:D.00}
\end{equation}
For simplicity, from now on, we drop the dependence of $\nu_t(\beta)$ and $\rho(\beta)$ on $\beta$.
For $t\geq0$ define the two sequences $\{\omega_t\},\{c_t\}$
$$
\omega_t=\Exp_\infty[\TS-t|\TS>t]=1+\sum_{n=t+1}^\infty\prod_{l=t+1}^n\Pro_{\infty,l}(\ell_{l}<\nu_{l})~\text{and}~c_{t}=\frac{\omega_{t+1}}{\nu_{t+1}}.
$$
Also set $c_{-1}=0$ and $\ell_0=0$.
From the definition of $\omega_t$ and comparing with \eqref{eq:th4.2} we conclude that $\omega_0=\gamma$.
Note that $\{\omega_t\}$ satisfies the backward recursion
\begin{equation}
\omega_{t-1}=1+\Pro_{\infty,t}(\ell_{t}<\nu_{t})\omega_{t}.
\label{eq:D.01}
\end{equation}
From \eqref{eq:D.00} we have $\Pro_{\infty,t}(\ell_t<\nu_t)\leq\rho$ suggesting that $\omega_t\leq1/(1-\rho)$. Furthermore
$$
1-\beta=\Pro_{0,t}(\ell_t<\nu_t)=\Exp_\infty[\ell_t\ind{\ell_t<\nu_t}]\leq\nu_t,
$$
from which we conclude that $c_t\leq1/(1-\beta)(1-\rho)$. In other words both sequences $\{\omega_t\},\{c_t\}$ are uniformly bounded from above by some finite constant.

Consider first \eqref{eq:th4.2}. The function $\phi(\beta)=1+\sum_{t=1}^\infty\prod_{l=1}^t\Pro_{\infty,l}\big(\ell_l<\nu_l(\beta)\big)$ is decreasing in $\beta$ with $\phi(0)=\infty$ and $\phi(1)=1$.
From assumption \eqref{eq:th4.1.2} we have validity of \eqref{eq:D.00} which allows for the use of Bounded Convergence to show that $\phi(\beta)$ is continuous in $\beta$. This suggests that \eqref{eq:th4.2} has a nonnegative solution.

As in the previous theorem we can write
$$
\ccJ_{\rm L}(T)\Pro_\infty(T>t)\leq\Exp_\infty[\ell_{t+1}\ind{T=t+1}].
$$
Multiplying both sides with $c_{t}$, which is nonnegative, and summing over $t\geq0$ we deduce that for any $T>0$ we have
$$
\ccJ_{\rm L}(T)\leq\frac{\Exp_\infty[\ell_Tc_{T-1}]}{\Exp_\infty[\sum_{t=0}^{T-1}c_t]}.
$$
Enlarging the class of stopping times $T$ by allowing randomization at time 0 with probability $\varpi$, recalling that $\ell_0=0$ and
using similar arguments as in the proof of Theorem\,\ref{th:3}, we can show that
\begin{equation}
\sup_{T:\Exp_\infty[T]\geq\gamma}\ccJ_{\rm L}(T)\leq\sup_{T:\Exp_\infty[T]=\gamma}\frac{\Exp_\infty[\ell_Tc_{T-1}]}{\Exp_\infty[\sum_{t=0}^{T-1}c_t]},
\label{eq:D.0}
\end{equation}
namely, to maximize the upper bound it suffices to limit ourselves to stopping times that satisfy the false alarm constraint with equality. We will show that the upper bound cannot exceed $\beta$.

Fix $T$ with $\Exp_\infty[T]=(1-\varpi)\Exp_\infty[T|T>0]=\gamma$ and consider the expression
\begin{multline}
\ccG(T)=\Exp_\infty\left[\ell_Tc_{T-1}-\beta\sum_{t=0}^{T-1}c_t+T\right]\\
=(1-\varpi)\Exp_\infty\left[\ell_Tc_{T-1}+\sum_{t=0}^{T-1}(1-\beta c_t)|T>0\right].
\label{eq:D.1}
\end{multline}
Note that $\varpi=1$ is not an acceptable value since then $T$ cannot satisfy the false alarm constraint with equality. Therefore $0\leq\varpi<1$. This suggests that $\Exp_\infty[T|T>0]=\gamma/(1-\varpi)<\infty$. We first examine the part $T>0$, namely the expression
$$
\hat{\ccG}(T)=\Exp_\infty\left[\ell_Tc_{T-1}+\sum_{t=0}^{T-1}(1-\beta c_t)|T>0\right].
$$
We observe that
\begin{multline*}
\Exp_\infty[\ell_T|T>0]=\sum_{t=1}^\infty\Exp_\infty[\ell_t\ind{T=t}|T>0]\leq\sum_{t=1}^\infty\Exp_\infty[\ell_t\ind{T>t-1}|T>0]\\
=\sum_{t=1}^\infty\Exp_\infty\left[\Exp_\infty[\ell_t|\ccF_{t-1}]\ind{T>t-1}|T>0\right]=\Exp_\infty[T|T>0]<\infty.
\end{multline*}
Since $\{c_t\}$ is uniformly bounded and because of the previous observation, this suggests that for every $\epsilon>0$ we can find sufficiently large integer $M$ so that $|\hat{\ccG}(T)-\hat{\ccG}(T_M)|\leq\epsilon$, where $T_M=\min\{T,M\}$. This implies
\begin{equation}
\hat{\ccG}(T)\leq\hat{\ccG}(T_M)+\epsilon.
\label{eq:D.2}
\end{equation}
We can now maximize $\hat{\ccG}(T_M)$ over $T_M$ with the optimization performed over the finite time horizon $[0,M]$. From standard optimal stopping theory we can define the sequence of optimal costs with the help of the backward recursion
$$
V_t(\ell_t)=\max\{\ell_tc_{t-1},(1-\beta c_{t})+\Exp_\infty[V_{t+1}(\ell_{t+1})]\};~t=M-1,\ldots,0,
$$
starting with $V_M(\ell_M)=\ell_Mc_{M-1}$. Since $V_M(\ell_M)\leq\max\{\ell_Mc_{M-1},\omega_M\}$, using induction we can show that $V_t(\ell_t)\leq\max\{\ell_tc_{t-1},\omega_t\}$ for all $t=M,M-1,\ldots,0$. Indeed, the inequality is true for $t=M$. Assume it is true for $t+1<M$, we will then prove it for $t$. Note that
\begin{multline*}
V_t(\ell_t)=\max\{\ell_tc_{t-1},(1-\beta c_{t})+\Exp_\infty[V_{t+1}(\ell_{t+1})]\}\\
\leq\max\{\ell_tc_{t-1},(1-\beta c_{t})+\Exp_\infty[\max\{\ell_{t+1}c_t,\omega_{t+1}\}]\}\\
=\max\{\ell_tc_{t-1},(1-\beta c_{t})+c_t\Pro_{0,t+1}(\ell_{t+1}\geq\nu_{t+1})+\omega_{t+1}\Pro_{\infty,t+1}(\ell_{t+1}<\nu_{t+1})\}\\
=\max\{\ell_tc_{t-1},1+\omega_{t+1}\Pro_{\infty,t+1}(\ell_{t+1}<\nu_{t+1})\}=\max\{\ell_tc_{t-1},\omega_{t}\}.
\end{multline*}
The inequality above is due to the induction assumption; furthermore, in the last three equalities we used the definition of $c_t$, namely, $c_{t}=\omega_{t+1}/\nu_{t+1}$; the fact that by construction of the sequence $\{\nu_t\}$ we have $\Pro_{0,t+1}(\ell_{t+1}\geq\nu_{t+1})=\beta$; and we also used recursion \eqref{eq:D.01}. We thus conclude that $V_t(\ell_t)\leq\max\{\ell_tc_{t-1},\omega_t\}$. Applying it for $t=0$ yields $V_0(\ell_0)\leq\max\{\ell_0c_{-1},\omega_0\}=\omega_0=\gamma$, because $\ell_0$ is defined to be 0 and, as we argued, $\omega_0=\gamma$. From optimal stopping theory we have $\hat{\ccG}(T_M)\leq V_0(\ell_0)$, consequently $\hat{\ccG}(T_M)\leq\gamma$. Using this in \eqref{eq:D.2} we obtain
$$
\hat{\ccG}(T)\leq\hat{\ccG}(T_M)+\epsilon\leq\gamma+\epsilon,
$$
which implies $\hat{\ccG}(T)\leq\gamma$. Substituting in \eqref{eq:D.1} and maximizing over $\varpi$, we have
$$
\ccG(T)\leq(1-\varpi)\gamma\leq\gamma,
$$
with the optimum randomization being $\varpi=0$. Using the definition of $\ccG(T)$ from \eqref{eq:D.1} and the fact that we consider $T$ with $\Exp_\infty[T]=\gamma$ we have
$$
\gamma\geq\ccG(T)=\Exp_\infty[\ell_Tc_{T-1}]-\beta\Exp_\infty\left[\sum_{t=0}^{T-1}c_t\right]+\Exp_\infty[T]$$
which directly implies
$$
\frac{\Exp_\infty[\ell_Tc_{T-1}]}{\Exp_\infty\left[\sum_{t=0}^{T-1}c_t\right]}\leq\beta.
$$
Shewhart, by construction, is an equalizer, hence we have $\ccJ_{\rm L}(\TS)=\beta$. From \eqref{eq:D.0} and the previous inequality we can then write
$$
\ccJ_{\rm L}(\TS)\leq\sup_{T:\Exp_\infty[T]\geq\gamma}\ccJ_{\rm L}(T)\leq\sup_{T:\Exp_\infty[T]=\gamma}
\frac{\Exp_\infty[\ell_Tc_{T-1}]}{\Exp_\infty\left[\sum_{t=0}^{T-1}c_t\right]}\leq\beta=\ccJ_{\rm L}(\TS),
$$
thus proving the desired optimality for Lorden's criterion. Similar proof applies in the case of Pollak's measure.\qed

\vskip0.2cm
\textsc{Proof of Theorem \ref{th:5}.}
When $q=0$ or $1$ then $\Pro_\infty(\ell_1^i\geq\nu)$ is continuous and strictly decreasing in $\nu$ (see Footnote\,\ref{foot:3}). If $q\in(0,1)$ we observe
\begin{equation}
\Pro_\infty\big((1-q)\ell_1^1+q\ell_1^2\geq\nu\big)=\int_0^\infty\Pro_\infty\left(\ell_1^1\geq\frac{\nu-qs}{1-q}\right)\Pro_\infty(\ell_1^2\in ds).
\label{eq:th5.10}
\end{equation}
Consequently if we use the continuity and strict monotonicity with respect to $\nu$ of the first probability under the integral and Bounded Convergence we can prove continuity and strict monotonicity of $\Pro_\infty((1-q)\ell_1^1+q\ell_1^2\geq\nu)$ as a function of $\nu$ for all $q\in[0,1]$. This probability is equal to 1 and 0 for $\nu=0$ and $\nu\to\infty$ respectively therefore there exists unique $\nu(q)\geq0$ that satisfies the false alarm constraint \eqref{eq:nuq} with equality.

Consider now $\nu(q)$ as a function of $q$. We like to show that this function is continuous. Fix $q_0\in(0,1)$ then for $q\to q_0\pm$ we will show $\nu(q_0\pm)=\nu(q_0)$. Recall that $\nu(q)$ is constructed so that for all $q\in[0,1]$ we have $\Pro_\infty((1-q)\ell_1^1+q\ell_1^2\geq\nu(q))=1/\gamma$.
Taking the limit with respect to $q\to q_0\pm$ and using \eqref{eq:th5.10} we have
\begin{multline*}
\frac{1}{\gamma}
=\lim_{q\to q_0\pm}\int_0^\infty\Pro_\infty\left(\ell_1^1\geq\frac{\nu(q)-qs}{1-q}\right)\Pro_\infty(\ell_1^2\in ds)\\
=\int_0^\infty\Pro_\infty\left(\ell_1^1\geq\frac{\nu(q_0\pm)-q_0s}{1-q_0}\right)\Pro_\infty(\ell_1^2\in ds)\\
=\Pro_\infty\big((1-q_0)\ell_1^1+q_0\ell_1^2\geq\nu(q_0\pm)\big),
\end{multline*}
where for the second equality we used Bounded Convergence and the continuity of the cdf of $\ell_1^1$. Since
$\Pro_\infty((1-q_0)\ell_1^1+q_0\ell_1^2\geq\nu(q_0\pm))=1/\gamma$ but also from the definition of $\nu(q_0)$ that $\Pro_\infty((1-q_0)\ell_1^1+q_0\ell_1^2\geq\nu(q_0))=1/\gamma$, we can claim that $\nu(q_0\pm)=\nu(q_0)$ because for each $q$, as we argued before, the threshold that satisfies the false alarm constraint with equality is unique. Similar proof (with one-sided limits) applies for $q_0=0,1$.

Let us now prove the validity of our theorem when the condition of case~i) is true. We have
\begin{multline}
\ccJ_{\rm L}(T)=\min_{i=1,2}\inf_{t\geq0}\esinf\Pro_t^i(T=t+1|\ccF_t,T>t)\\
\leq\inf_{t\geq0}\esinf\Pro_t^1(T=t+1|\ccF_t,T>t)\\
\leq\inf_{t\geq0}\esinf\Pro_t^1(\TS(0)=t+1|\ccF_t,\TS(0)>t)=\Pro_0^1\big(\ell_1\geq\nu(0)\big),
\label{eq:E1}
\end{multline}
where the second inequality comes from the fact that $\TS(0)$ is optimum when the post-change probability measure is $\Pro_0^1$, and the last equality is the result of $\TS(0)$ being an equalizer under $\Pro_0^1$. Note now that
$$
\Pro_0^1\big(\ell_1\geq\nu(0)\big)\leq\Pro_0^2\big(\ell_1\geq\nu(0)\big)=\inf_{t\geq0}\esinf\Pro_t^2\big(\TS(0)=t+1|\ccF_t,\TS(0)>t\big),
$$
the inequality being the condition of case~i) and the equality that follows is the result of $\TS(0)$ being an equalizer under $\Pro_0^2$ as well. Completing what was started in \eqref{eq:E1}, we can write
\begin{multline*}
\ccJ_{\rm L}(T)\leq\Pro_0^1\big(\ell_1\geq\nu(0)\big)=\min_{i=1,2}\Pro_0^i\big(\ell_1\geq\nu(0)\big)\\
=\min_{i=1,2}\inf_{t\geq0}\esinf\Pro_t^i\big(\TS(0)|\ccF_t,\TS(0)>t\big)=\ccJ_{\rm L}\big(\TS(0)\big),
\end{multline*}
which proves the claim of case~i). Similar proof applies in case~ii).

Suppose now that neither the condition of case~i) nor of case~ii) is valid. This suggests that we simultaneously have $\Pro_0^2(\ell_1^1\geq\nu(0))<\Pro_0^1(\ell_1^1\geq\nu(0))$ and $\Pro_0^1(\ell_1^2\geq\nu(1))<\Pro_0^2(\ell_1^2\geq\nu(1))$. Define the following difference as a function of $q$
$$
{\sf D}(q)=\Pro_0^1\big((1-q)\ell_1^1+q\ell_1^2\geq\nu(q)\big)-\Pro_0^2\big((1-q)\ell_1^1+q\ell_1^2\geq\nu(q)\big).
$$
We observe that ${\sf D}(0)>0$ and ${\sf D}(1)<0$, furthermore ${\sf D}(q)$ is continuous because we can show using \eqref{eq:th5.10} and the continuity of $\nu(q)$ that the probabilities $\Pro_0^i((1-q)\ell_1^1+q\ell_1^2\geq\nu(q))$ are continuous in $q$. Hence there exists $q\in(0,1)$ so that ${\sf D}(q)=0$.
For this specific $q$ the corresponding Shewhart stopping rule $\TS(q)$ is by construction an equalizer across time and across post-change probabilities. Furthermore for each $T$ and $t\geq0$, as in \eqref{eq:C.0}, we have
$$
\ccJ_{\rm L}(T)\Pro_\infty(T>t)\leq\Pro_t^i(T=t+1);~~i=1,2
$$
suggesting
\begin{multline*}
\ccJ_{\rm L}(T)\Pro_\infty(T>t)\leq(1-q)\Pro_t^1(T=t+1)+q\Pro_t^2(T=t+1)\\
=\Exp_\infty[\{(1-q)\ell_{t+1}^1+q\ell_{t+1}^2\}\ind{T=t+1}].
\end{multline*}
Summing over $t\geq0$ we obtain the following upper bound
$$
\ccJ_{\rm L}(T)\leq\frac{\Exp_\infty[(1-q)\ell_T^1+q\ell_T^2]}{\Exp_\infty[T]}.
$$
The proof continues along the same lines of the proof of Theorem\,\ref{th:3}. Basically we show that the upper bound is optimized by $\TS(q)$, furthermore this optimal value is also attained by $\ccJ_{\rm L}(\TS(q))$ because $\TS(q)$ is an equalizer across time and across post-change probabilities. This establishes the desired optimality for $\TS(q)$. 

What is now left to demonstrate is that for each $\gamma$, only one of the three cases can be valid. Call
$$
\ccJ_{\rm L}^i(T)=\inf_{t\geq0}\esinf\Pro_t^i(T=t+1|\ccF_t,T>t);~~i=1,2,
$$
then we know that $\ccJ_{\rm L}^1(T)$ is maximized by $\TS(0)$ and $\ccJ_{\rm L}^2(T)$ by $\TS(1)$. In fact no other stopping time can attain the same optimal value unless it is equal, with probability 1, to the corresponding Shewhart test. If case~i) applies then we will show that it is not possible the condition of case~ii) to be true. Indeed, if both conditions were valid simultaneously, then we could write
\begin{multline*}
\ccJ_{\rm L}^1\big(\TS(0)\big)=\Pro_0^1\big(\ell_1^1\geq\nu(0)\big)\leq\Pro_0^2\big(\ell_1^1\geq\nu(0)\big)=\ccJ_{\rm L}^2\big(\TS(0)\big)\\
\leq\ccJ_{\rm L}^2\big(\TS(1)\big)=\Pro_0^2\big(\ell_2^1\geq\nu(1)\big)\leq\Pro_0^1\big(\ell_2^1\geq\nu(1)\big)=\ccJ_{\rm L}^1\big(\TS(1)\big),
\end{multline*}
where the first inequality comes from case~i), the second inequality from the fact that $\TS(1)$ optimizes $\ccJ_{\rm L}^2(T)$ and the third inequality is the condition of case~ii). From the above we conclude that $\TS(1)$ has a better $\ccJ_{\rm L}^1(\cdot)$ performance than $\TS(0)$ which optimizes $\ccJ_{\rm L}^1(\cdot)$, leading to contradiction. Actually since $\TS(1)$ is not equal to $\TS(0)$ with probability 1, its corresponding performance is strictly smaller than the optimum. 

Similarly it is not possible to have the conditions of case~i) and case~iii) be satisfied at the same time. Again if this were true then
\begin{multline*}
\ccJ_{\rm L}^1\big(\TS(0)\big)=\Pro_0^1\big(\ell_1^1\geq\nu(0)\big)\leq(1-q)\Pro_0^1\big(\ell_1^1\geq\nu(0)\big)+q\Pro_0^2\big(\ell_1^1\geq\nu(0)\big)\\
=(1-q)\ccJ_{\rm L}^1\big(\TS(0)\big)+q\ccJ_{\rm L}^2\big(\TS(0)\big)\\
\leq(1-q)\ccJ_{\rm L}^1\big(\TS(q)\big)+q\ccJ_{\rm L}^2\big(\TS(q)\big)=\ccJ_{\rm L}^1\big(\TS(q)\big),
\end{multline*}
with the first inequality due to case~i) and the second due to the fact that the convex combination of the two measures is maximized by $\TS(q)$. Finally the last equality is true because of case~iii) namely that the stopping time $\TS(q)$ is an equalizer for the two post-change measures. Again this is a contradiction since $\TS(q)$ has larger $\ccJ_{\rm L}^1(\cdot)$ measure than $\TS(0)$ which is the optimum. Therefore case~i) and case~iii) cannot be valid at the same time. Similarly we can show that case~ii) and case~iii) are incompatible.

Since we have shown that when neither case~i) nor case~ii) is valid, we necessarily have case~iii) being true, this suggests that, for each value of $\gamma$, exactly one of the three cases applies. This concludes the proof for Lorden's criterion. Similar proof applies in the case of Pollak's measure. \qed

\vskip0.2cm
\textsc{Proof of Theorem \ref{th:6}.} Let case~i) be true, then we can write
\begin{multline}
\cJ_{\rm L}(T)\geq\sup_{t\geq0}\esup\Exp_t^1[T-t|\ccF_t,T>t]\\
\geq\sup_{t\geq0}\esup\Exp_t^1[\TS(0)-t|\ccF_t,\TS(0)>t]=\Exp_0^1[\TS(0)]=\frac{1}{\Pro_0^1(\ell_1\geq\nu(0))},
\label{eq:F.1}
\end{multline}
where the first inequality is obvious and the second comes from the fact that if $\nu(0)\leq1$ then the Shewhart stopping time $\TS(0)$, according to Section\,\ref{ssec:1.1}, optimizes Lorden's original criterion for the post-change probability measure $\Pro_0^1$. The second last equality comes from the fact that Shewhart, exactly as CUSUM, is an equalizer and the last equality is true due to \eqref{eq:Average}. We also have
$$
\sup_{t\geq0}\esup\Exp_t^2[\TS(0)-t|\ccF_t,\TS(0)>t]=\Exp_0^2[\TS(0)]=\frac{1}{\Pro_0^2\big(\ell_1\geq\nu(0)\big)},
$$
because $\TS(0)$ is an equalizer under $\Pro_0^2$ as well. Since by assumption, $\Pro_0^1(\ell_1\geq\nu(0))\leq\Pro_0^2(\ell_1\geq\nu(0))$ this suggests that $\cJ_{\rm L}(\TS(0))=\max_{i=1,2}1/\Pro_0^i(\ell_1\geq\nu(0))=1/\Pro_0^1(\ell_1\geq\nu(0))$. Using this last observation in \eqref{eq:F.1} we conclude that $\cJ_{\rm L}(T)\geq\cJ_{\rm L}(\TS(0))$, thus proving optimality of $\TS(0)$. In a similar way we can prove optimality for $\TS(1)$ under the condition of case~ii).

Assume now that we are in case~iii) then
\begin{multline*}
\cJ_{\rm L}(T)\geq\Exp_t^i[T-t|\ccF_t,T>t]=\Exp_t^i\left[\sum_{n=t}^\infty\ind{T>n}|\ccF_t,T>t\right]\\
=\sum_{n=t}^\infty\Exp_t^i[\ind{T>n}|\ccF_t,T>t]=\sum_{n=t}^\infty\Exp_\infty\left[\ind{T>n}\prod_{m=t+1}^n\ell_m^i|\ccF_t,T>t\right]\\
\Exp_\infty\left[\sum_{n=t}^{T-1}\prod_{m=t+1}^n\ell_m^i|\ccF_t,T>t\right],
\end{multline*}
where we applied a change of measures and used the fact that $\{T>n\}$ is $\ccF_n$-measurable. We also define $\prod_{a}^b=1$ and $\sum_{a}^b=0$ when $b<a$ while we recall that $\ell_0^i$ is defined to be 0. Multiplying both sides of the previous inequality with $\ind{T>t}(1-\ell_t^i)^+$ which is nonnegative and $\ccF_t$-measurable and taking expectation with respect to the nominal measure, we obtain
\begin{multline*}
\cJ_{\rm L}(T)\Exp_\infty[\ind{T>t}(1-\ell_t^i)^+]\geq
\Exp_\infty\left[\sum_{n=t}^{T-1}\ind{T>t}\prod_{m=t+1}^n\ell_m^i(1-\ell_t^i)^+\right]\\
\geq\Exp_\infty\left[\sum_{n=t}^{T-1}\ind{T>t}\prod_{m=t+1}^n\ell_m^i(1-\ell_t^i)\right]\\
=\Exp_\infty\left[\sum_{n=t}^{T-1}\ind{T>t}\left(\prod_{m=t+1}^n\ell_m^i-\prod_{m=t}^n\ell_m^i\right)\right].
\end{multline*}
Summing over all $t\geq0$ and recalling that $\ell_0^i=0$, $\prod_{n+1}^n=1$, yields
\begin{multline*}
\cJ_{\rm L}(T)\Exp_\infty\left[\sum_{t=0}^{T-1}(1-\ell_t^i)^+\right]
\geq\Exp_\infty\left[\sum_{t=0}^{T-1}\sum_{n=t}^{T-1}\left(\prod_{m=t+1}^n\ell_m^i-\prod_{m=t}^n\ell_m^i\right)\right]\\
=\Exp_\infty\left[\sum_{n=0}^{T-1}\sum_{t=0}^{n}\left(\prod_{m=t+1}^n\ell_m^i-\prod_{m=t}^n\ell_m^i\right)\right]=\Exp_\infty\left[\sum_{n=0}^{T-1}1\right]=\Exp_\infty[T].
\end{multline*}
Finally multiplying the previous inequality for $i=1$ with $(1-q)$ and the one for $i=2$ with $q$ and adding the resulting expressions we obtain the following lower bound
$$
\cJ_{\rm L}(T)\geq\frac{\Exp_\infty[T]}{\Exp_\infty\left[\sum_{t=0}^{T-1}(1-q)(1-\ell_t^1)^+ +q(1-\ell_t^2)^+\right]}.
$$

Following the usual methodology we have adopted in the previous proofs, in order to minimize the lower bound, with the help of the randomization probability $\varpi$ we can show that we can limit ourselves to stopping times that satisfy the false alarm constraint with equality. Consequently
\begin{multline}
\inf_{T:\Exp_\infty[T]\geq\gamma}\cJ_{\rm L}(T)\\
\geq\inf_{T:\Exp_\infty[T]=\gamma}\frac{\Exp_\infty[T]}{\Exp_\infty\left[\sum_{t=0}^{T-1}(1-q)(1-\ell_t^1)^+ +q(1-\ell_t^2)^+\right]}\\
=\frac{\gamma}{\displaystyle\sup_{T:\Exp_\infty[T]=\gamma}\Exp_\infty\left[\sum_{t=0}^{T-1}(1-q)(1-\ell_t^1)^+ +q(1-\ell_t^2)^+\right]}.
\label{eq:F.2}
\end{multline}
For simplicity denote $z_t=(1-q)(1-\ell_t^1)^+ +q(1-\ell_t^2)^+$, then maximizing the denominator subject to the equality constraint is straightforward. Using a Lagrange multiplier with value 
$$
\lambda=(1-\nu)\Pro_\infty(z_1<1-\nu)+\Exp_\infty[z_1\ind{z_1\geq1-\nu}]
$$
and applying standard optimal stopping theory, we can conclude that the optimum stopping time is
$$
\mathcal{T}=\inf\{t>0:z_t\leq1-\nu\}.
$$

For $\nu=\nu(q)$ we will show that $\mathcal{T}$ is in fact equivalent to $\TS(q)$ under condition \eqref{eq:th6.2}. Indeed notice that when $\mathcal{T}$ stops we have
$$
1-\nu(q)\geq z_{\mathcal{T}}=(1-q)(1-\ell_\mathcal{T}^1)^+ +q(1-\ell_\mathcal{T}^2)^+\geq(1-q)(1-\ell_\mathcal{T}^1) +q(1-\ell_\mathcal{T}^2)
$$
which implies
$$
(1-q)\ell_\mathcal{T}^1 +q\ell_\mathcal{T}^2\geq\nu(q),
$$
suggesting $\TS(q)\leq\mathcal{T}$ (because $\TS(q)$ is the \textit{first} time instant the above inequality is true). For any $t<\mathcal{T}$ we have
\begin{equation}
1-\nu(q)<z_t=(1-q)(1-\ell_t^1)^+ +q(1-\ell_t^2)^+.
\label{eq:app100}
\end{equation}
Because $\nu(q)$ satisfies \eqref{eq:th6.2} we will show that the previous inequality can be true only when $\xi_t\in\cA_1\cap\cA_2$, that is, when the likelihood ratios $\ell_t^1$ and $\ell_t^2$ are simultaneously no larger than 1. Indeed from \eqref{eq:th6.2} we have that the upper bound of $\nu(q)$ is no larger than 1, consequently in \eqref{eq:app100} the two likelihood ratios cannot be larger than 1 simultaneously. Let $\ell_t^1\leq1$ and $\ell_t^2>1$ then \eqref{eq:app100} becomes $1-\nu(q)<(1-q)(1-\ell_t^1)$ or $q+(1-q)\ell_t^1<\nu(q)$. But the latter is again not possible because of the left hand side inequality of \eqref{eq:th6.2}. The same is true when $\ell_t^2\leq1$ and $\ell_t^1>1$. Hence \eqref{eq:app100} can be valid only when both likelihood ratios are smaller than 1. This means that when $t<\mathcal{T}$, \eqref{eq:app100} is equivalent to
$$
(1-q)\ell_t^1 +q\ell_t^2<\nu(q).
$$
This observation suggests that $t<\mathcal{T}$ combined with \eqref{eq:th6.2} implies $t<\TS(q)$, therefore $\mathcal{T}-1<\TS(q)$ or $\TS(q)\geq\mathcal{T}$. Consequently $\TS(q)=\mathcal{T}$, which means that $\TS(q)$ optimizes the lower bound in \eqref{eq:F.2}. 

To compute the optimum value of the lower bound, since before stopping both likelihood ratios are no larger than 1, we note
\begin{multline*}
\hskip-0.3cm\Exp_\infty\left[\sum_{t=0}^{\TS(q)-1}(1-q)(1-\ell_t^1)^+ +q(1-\ell_t^2)^+\right]
=\gamma-\Exp_\infty\left[\sum_{t=0}^{\TS(q)-1}(1-q)\ell_t^1 +q\ell_t^2\right]\\
=\gamma\big\{(1-q)\Pro_0^1\big((1-q)\ell_t^1 +q\ell_t^2\geq\nu(q)\big)+q\Pro_0^2\big((1-q)\ell_t^1 +q\ell_t^2\geq\nu(q)\big)\big\}\\
=\gamma\Pro_0^1\big((1-q)\ell_t^1 +q\ell_t^2\geq\nu(q)\big)=\gamma\Pro_0^2\big((1-q)\ell_t^1 +q\ell_t^2\geq\nu(q)\big),
\end{multline*}
where we used the fact that $\gamma=\Exp_\infty[\TS(q)]=1/\Pro_\infty((1-q)\ell_t^1 +q\ell_t^2\geq\nu(q))$
and that we are in case~iii) with condition \eqref{eq:th5.1} being valid.
Consequently
$$
\inf_{T:\Exp_\infty[T]\geq\gamma}\cJ_{\rm L}(T)
\geq\frac{1}{\Pro_0^i\big((1-q)\ell_t^1 +q\ell_t^2\geq\nu(q)\big)}.
$$
Now it is straightforward to verify that the lower bound is attainable by the Lorden measure of the Shewhart stopping time $\TS(q)$. This is clearly due to the fact that $\TS(q)$ is an equalizer across time and across post-change measures. This concludes our proof.\qed

\end{document}